\documentclass{amsproc}
\usepackage{amsmath,amssymb,amsthm,amsfonts,latexsym}
\usepackage{txfonts}
\makeatletter
\newcommand\footnoteref[1]{\protected@xdef\@thefnmark{\ref{#1}}\@footnotemark}

\theoremstyle{plain}
\newtheorem{thm}{Theorem}
\newtheorem{pro}[thm]{Proposition}
\newtheorem{lem}[thm]{Lemma}
\newtheorem{cor}[thm]{Corollary}

\newtheorem{facts}[thm]{Fact}

\theoremstyle{definition}

\theoremstyle{remark}
\newtheorem{rem}[thm]{{\it Remark}}

\newcounter{dc}

\DeclareMathOperator{\okr}{{\stackrel{{\scriptscriptstyle{\mathsf{def}}}}{=}}}

\DeclareMathOperator{\D}{d\!}

\DeclareMathOperator{\E}{e} \DeclareMathOperator{\I}{i}
   
   \DeclareMathOperator{\IM}{\mathfrak{Im}}

\DeclareMathOperator{\lin}{lin}

\DeclareMathOperator{\clolin}{clolin}

\def\dz#1{\mathcal D({#1})}

\def\funk#1#2#3{#1\colon#2\to#3}

\def\Ge{\geqslant}

\def\is#1#2{\langle#1,#2\rangle}
\def\jd#1{\mathcal N(#1)}

\def\Le{\leqslant}
\def\liczp#1{{${#1}^{\text {\rm o}}$}}

\def\ob#1{\mathcal R(#1)}

\def\poch#1#2#3{\frac{\D{\!\,}^{#3}{#1}}{\D{#2}^{#3}}}

\def\res#1{|_{#1}}
\def\rres#1{\!\!\upharpoonright_{#1}}

\def\sbar#1{\,\overline{\!#1}}

\def\ulamek#1#2{\mbox{\normalfont$\frac{#1}{#2}$}}

\def\zb#1#2{\{{#1}\colon\ {#2}\}}

\def\aac{\mathcal A}

\def\bbc{\mathcal B}

\def\ccc{\mathcal C}

\def\ddc{\mathcal D}

\def\eec{\mathcal E}

\def\hhc{\mathcal H}

\def\kkc{\mathcal K}

\def\llc{\mathcal L}
\def\qqc{\mathcal Q}

\def\ccb{\mathbb C}

\def\rrb{\mathbb R}

\def\ggf{\mathfrak G}

\begin{document}

\title[Squeezing of arbitrary order]
{Squeezing of arbitrary order: the ups and downs}
\author{
Katarzyna G\'{o}rska$^{1}$, Andrzej Horzela$^{1}$ and Franciszek Hugon Szafraniec$^{2}$}

\address{$^{1}$H. Niewodnicza\'{n}ski Institute of Nuclear Physics, Polish Academy of Sciences, Division of Theoretical Physics, ul. Eliasza-Radzikowskiego 152, PL 31-342 Krak\'{o}w, Poland\\
$^{2}$Instytut Matematyki, Uniwersytet Jagiello\'{n}ski, ul. \L ojasiewicza 6,
PL 30 348 Krak\'ow, Poland}




\begin{abstract}
We show how using classical von Neumann index theory makes it possible a universal treatment of squeezing of arbitrary order. ``Universal'' means that the same approach applied to displacement (order $1$) and squeeze (order $2$) operators confirms toughly what is already known as well as provides rigorous arguments that the higher order squeezing can not be generalized in a  ``naive'' way. We create an environment for answering \underbar{definitely} all the emerging questions 
in positive ({\bf the ups}) and negative ({\bf downs}). In the latter case we suggest ways for further development.

\end{abstract}
\maketitle

In the eighties a tendency to generalize squeezed states, and squeeze operators in particular, to higher orders became present in the  literature  \cite{brau,nieto}. The authors 
discussed, not expected by physicists, impossibility of exponentiating the operators $A_{\xi}^{(k)}=i{\xi}^{*}a_{-}^{k}- i{\xi}a_{+}^{k}$, $k\ge 3$ basing their arguments on  
showing non-analycity of the vacuum state. The latter is however not decisive for the lack of selfadjointness of those $A_{\xi}^{(k)}$s and creates a problem to be explained, cf. \cite{nagel}. The basic requirement therein, namely normalizability of squeezed states defined via the Bogolubov transform of $ a_{-}^{k}$, turns out to be misleading if $k\ge 3$. In the recent   paper \cite{yzz} solutions to the Schr\"odinger equation for squeezed harmonic oscillators, considered in the Segal-Bargmann space, have been shown to be non-normalizable for $k\ge 3$.  
This not only remains in contradiction to what is in \cite{nagel} but also confirms
earlier findings for similar $k$-photon Rabi model with the interaction $\sigma_{x}(g^{*}a_{-}^{k}+ga_{+}^{k})$ where $\sigma_{x}$ is the Pauli 
matrix. It has been known for several years \cite{cflo98} that this model suffers, for
$k\ge 3$, from analogous pathologies  as the generalized squeezing does.
Discussion of the above, recently quite extensive \cite{yzz,dajka1,dajka2,comments,com2, gar,lo}, does not get 
rid of  difficulties arisen 
as   the crucial question
of  selfadjointness of  $A_{\xi}^{(k)}$,
$k\ge 3$, is left untouched. An attempt at compensating the lack of selfadjointness   with  suitable modifications, like  selfadjoint extensions, may lead to
physically important consequences.


The aim of our paper is to provide both communities, physicists and mathematicians, with adequate grounds for settling the appearing inconsistency. As a kind of surprise the main tool which works perfectly for this purpose turns out to be very classical and it is nothing but the von Neumann deficiency index approach. It makes the answers definite although reached after rather laborious calculations which we are presenting in detail so as to maintain mathematical rigor and to encourage others to follow.

We begin with preliminary notions to fix the language to be used. The main tool is to investigate essential selfadjointnesss of the operators $A_{\xi}^{(k)}$. We show that though the analytic vectors approach works well for $k=1,2$ is not sufficient to judge the problem for $k\Ge 3$. It is the von Neumann index theory which covers both cases giving definite answers: affirmative for $k=1,2$ and negative for $k\ge 3$. This streghtens universality of the apparatus we have chosen; all this is contained in Sections 2 and 3. Section 4 is devoted to analysis of possible subtleness' appearing in the process of exponentiation of  $A_{\xi}^{(k)}$ and related operators. In section 5 we go back to the case $k=1,2$ modeling them comprehensively in the Segal-Bargmann space. The paper is completed by concluding remarks in which we sum up its mathematical  aspects as well as briefly discuss their physical consequences and the Appendix containing a substantial part of calculations needed in the Section 2.

\section{Preliminaries}
\subsection{Basic notions}\label{basic1}
Let $\hhc$ be a Hilbert space. For an operator $A$ in $\hhc$,  $\dz A$ denotes its domain, $\ob A$ its range and $\jd A$ its null space (the kernel). $\sbar A$ stands always for  the closure
 of a closable operator $A$ and $A^{*}$ for its Hilbert space adjoint.

 If $\ddc\subset\dz A$ then the operator $A\res\ddc$ defined by $\dz{A\res\ddc}\okr \ddc$ and ${A\res\ddc}f\okr Af$, $f\in\ddc$, can be viewed as a {\em restriction} of $A$ to $\ddc$ and $A$ can be considered as an {\em extension} of $A\res\ddc$; this is a standard set theoretical notion. If $\ddc$ is dense in $\hhc$ then both $A$ and $A\res\ddc$ are a densely defined operators in $\hhc$  (that is $A\res\ddc$ acts within the same space $\hhc$ as $A$ does).

A linear subspace
 $\ddc$ of $\dz A$ is said to be a {\em core} of a closable operator $A$ if $\overline{A\res\ddc}=\sbar A$.

Furthermore, a subspace
$\ddc\subset\dz A$ is said to be {\it invariant}
for $A$ if $A\ddc\subset\ddc$. If this happens,   $A\res\ddc$ can also be thought of as a densely defined operator acting in the Hilbert space $\sbar \ddc$, the closure of $\ddc$. Again if $\ddc$ is dense in $\hhc$ the only difference between this and the previous case is  that in the latter $\ob{A}\subset\ddc$.
 
  On  the  other  hand,  
  a  \underbar{closed} subspace $\mathcal L$ of $\hhc$ is called  {\it invariant}
    for $A$  if
$A(\mathcal L\cap\dz A) \subset\llc$; then the  {\em restriction}
$A\rres{\mathcal L}\okr A\res{\mathcal L\cap\dz A}$ is always considered as an operator in $\llc$. If $\ddc$ in the above is closed then $A\res\ddc=A\rres\ddc$. A step
further, a closed subspace ${\mathcal L}$ {\it  reduces}  an
operator $A$ if both $\mathcal L$ and $\mathcal L^\perp$ are
invariant for $A$
as well as $P\dz A\subset\dz A$, where $P$ is the orthogonal
projection of $\hhc$ onto $\mathcal L$; all this is the
same  as  to  require  $P  A\subset  AP$.  If this happens, the   restriction
$A\rres{\mathcal L}$ is  called  a  {\it  part}  of  $A$  in
$\mathcal L$. If $\llc$ reduces $A$ and $A$ is densely defined then so is $A\rres\llc$.

Notice the word ``invariant'' has double meaning here but the circumstances we use it protect us from any confusion. 

\subsection{The operators}\label{s1.10.11}
Now let $\hhc$ be  a separable  Hilbert space (with the inner product to be linear in the first varaiable)
 and $(e_{n})_{n=0}^{\infty}$ be an orthonormal basis (i.e. an orthonormal complete set) 
in it\,\footnote{\; We follow the customary mathematical notation. In order to make our mathematical reasoning more clear we intentionally do not use the Fock space notation. We  believe it does not cause any problem for the physicists.}. 
The (abstract) {\em creation}  and {\em annihilation} operators  (with respect to the orthonormal basis $(e_{n})_{n=0}^{\infty}$) are  linearly extended from
 \begin{align}\label{a}
   \begin{gathered}\dz{a_{+}}=\dz{a_{-}}=\ddc\okr\lin\big(e_{n}\big)_{n=0}^{\infty},
   \\ 
   a_{+}e_{n}\okr\sqrt{n+1}\,e_{n+1},\quad n=0,1,\ldots\qquad\text{(creation)}
\\
a_{-}e_{n}\okr\sqrt{n}\,e_{n-1},\quad n=1,\ldots, \quad a_{-}e_{0}\okr0\qquad\text{(annihilation)}.
   \end{gathered}
   \end{align}

With the definitions \eqref{a} we sort out  selfadjointness of the  operators 
\begin{equation*}
A_{\xi}^{(k)}\okr\I\xi^{*} a_{-}^{k}-\I \xi a_{+}^{k},\quad k=0,1, \ldots
\end{equation*}
 with $\xi$
 being a complex parameter. As our ultimate goal is to prove \eqref{1.15.02}, $|\xi|$ has no impact on the problem and we drop it considering instead just the operators
 \begin{gather*}
A^{(k)}\okr {-}\I\big(\E^{\I\theta}a_{+}^{k}-\E^{-\I\theta}a_{-}^{k}\big), \quad\dz{A_{k}}\okr\ddc
\quad k=0,1, \ldots 
\end{gather*}
 with $\theta$ being a (fixed) real parameter. Therefore
 \begin{equation*}
  \text{if $\xi=\E^{\I\theta}|\xi|$, $A^{(k)}_{\xi}=|\xi |A^{(k)}$ \;and\; $A^{(k)}=A^{(k)}_{\exp[\I\theta]}$.}
 \end{equation*}
 The operators $a_{+}$ and $a_{-}$ are \underbar{formally} adjoint each to the other, that is
 \begin{equation*}
 \is {a_{+}f}{g}=\is{f}{a_{-}g},\quad f,g\in\ddc;
 \end{equation*}
 in physical tradition this fact is nicknamed as ``Hermitian adjoint'' and symbolized by ${}^{\dag}$, which makes some sense as long as the $a_{+}$ and ${a_{-}}$ are formal algebraic objects and no domain is indicated. This means that
\begin{equation*}
a_{+}=(a_{-})^{*}\res\ddc,\quad a_{-}=(a_{+})^{*}\res\ddc
\end{equation*}
as $\ddc$ is invariant for both $(a_{+})^{*}$ and $(a_{-})^{*}$.
Consequently, the operators $A^{(k)}$ are symmetric.

Moreover, it is a matter of direct calculation that $\ddc$ is a core of $(a_{+})^{*}$ and $(a_{-})^{*}$, and that for the closure one has
\begin{equation*}
\overline{a_{+}}=(a_{-})^{*},\quad \overline{a_{-}}=(a_{+})^{*}.
\end{equation*} 
Notice that
by means of the basis $(e_{n})_{n=0}^{\infty}$
\begin{equation}\label{2.16.06}
A^{(k)}e_{n}= {-}\I\left\{\E^{\I\theta} \sqrt{\ulamek{\big(n+k\big)!}{n!}}\,e_{n+k}-\E^{-\I\theta} \sqrt{\ulamek{n!}{\big(n-k\big)!}}\,e_{n-k}\right\}
\end{equation}
with notation $e_{-k}=e_{-k+1}=\ldots=e_{-1}=0$.
								
Defining\begin{equation}\label{1.16.06}
\ddc^{(k,i)}\okr\lin \big(e_{i+pk}\big)_{p=0}^{\infty},\quad
\hhc^{(k,i)}\okr\clolin \big(e_{i+pk}\big)_{p=0}^{\infty},\;\quad i=0,\ldots,k-1
\end{equation}
it is clear that $\ddc=\bigoplus_{i=0}^{k-1}\ddc^{(k,i)}$ and $\hhc=\bigoplus_{i=0}^{k-1}\hhc^{(k,i)}$. It is a kind of straightforward argument to verify the following.
\begin{pro}\label{1t.15.06}
Each $\hhc^{(k,i)}$ reduces $A^{(k)}$ and the domains $\dz{A^{(k,i)}}=\ddc_{i}$ are invariant for   the parts $A^{(k,i)}\okr\rres{\hhc^{(k,i)}}$ of $A^{(k)}$ in $\hhc^{(i)}$ as well as  the operators $A^{(k,i)}$  are symmetric. $A_{k}$ is essentially selfadjoint\,\footnote{\;An operator is called {\em essentially selfadjoint} if its closure is selfadjoint.} if and only if so is each $A^{(k,i)}$.
\end{pro}
Proposition \ref{1t.15.06} allows to downgrade the search for essential selfadjointness of $A^{(k)}_{\xi}$ to that of any of $A^{(k,i)}$'s. With the notation
$e^{( {k}, i)}_{p}\okr e_{i+pk}$, $p=0,1,\ldots$, the formula \eqref{2.16.06} reads as
\begin{equation}\label{1.18.06}
A^{(k,i)}e_{p}^{(k, i)}=-\I\left(\E^{\I\theta} \sqrt{\ulamek{[i+(p+1)k]!}{(i+pk)!}} \,e^{(k, i)}_{p+1} - \E^{-\I\theta}\sqrt{\ulamek{(i+pk)!}{[i+(p-1)k]!}}\,e^{(k, i)}_{p-1}\right).
\end{equation}
Notice that $\big(e_{p}^{(k,i)}\big)_{p=0}^{\infty}$ is an orthonormal basis in $\hhc^{(k,i)}$. The operators $A^{(k,i)}$ act as  Jacobi operators\,\footnote{\;A Jacobi operator in a separable Hilbert space acts according to the same (tridiagonal) pattern with respect to a chosen orthonormal basis as Jacobi matrix does (when considered in $\ell^{2}$ with respect to the zero-one basis; in this case it has to be considered in principle as an unbounded operator defined on ``finite'' vectors).} in $\hhc^{(k,i)}$ with zero diagonal. Therefore their deficiency indices are either $(0,0)$ or $(1,1)$, their representing measures are always symmetric with respect to $0$. Notice that each $A^{(k,i)}$ is a \underbar{cyclic} operator with a cyclic vector $e^{(k,i)}_{0}$, that is $\dz{A^{(k,i)}}=\lin\big(e^{(k,i)}_{p}\big)_{p=0}^{\infty}$. In conclusion,

\begin{cor}\label{t1.7.07}
 $A^{(k)}$ becomes an orthogonal sum $\bigoplus_{i=0}^{k-1}A^{(k,i)}$ of Jacobi operators with respect to the bases $\big(e_{p}^{(k,i)}\big)_{p=0}^{\infty}$, $i=0,\ldots k-1$, which are in particular cyclic.
 \end{cor}

 \section{Essential selfadjointness of the operators $A^{(k,i)}$; the first attempt - biased}\label{s2.22.09}
\subsection{$\ccc^{\infty}$-vectors. A recollection}\label{s1.30.11}
 Recall that, in general, $f\in\ddc^{\infty}(A)\okr\bigcap_{n=0}^{\infty}\dz{A^{n}}$ is 
 \begin{gather}
 \text{a {\em bounded} vector if there are $a>0$ and $b>0$ such that $\|A^{n}f\|\Le ab^{n}$ for $n=0,1,\dots$;}\notag
 \\
 \text{an {\em analytic} vector of $A$ if there is $t>0$ such that $\sum\nolimits_{n=0}^{\infty}\frac{t^{n}}{n!}\|A^{n}f\|<+\infty$;}\notag
 \\ \text{
 an {\em entire} vector if the convergence in the above holds for all $t>0$;}
\notag
\\
 \text{a {\em quasianalytic} vector of $A$ if }\sum_{n=0}^{\infty}\|A^{n}f\|^{-1/n}=+\infty.\notag
 \end{gather}
 All those vectors are customarily called $\ccc^{\infty}$-vectors of $A$. Let us introduce the self-evident notation  $\bbc(A)$, $\aac(A)$, $\eec(A)$ and $\qqc(A)$ for the consecutive classes. The first three are always linear while the last may not be. Nevertheless the inclusions $ \bbc(A)\subset\eec(A)\subset\aac(A)\subset\qqc(A)$ are transparent.
 It  may happen that even $\qqc(A)$ is a zero space. However when (essential) selfadjointness is around their nontriviality becomes essential.
 \begin{facts}\label{t3.30.11}
$($a$)$ If $A$ is selfadjoint then $\bbc(A)$ $($consequently,  $\aac(A)$, $\eec(A)$ and $\lin\qqc(A)$$)$ constitute a core of $A$.

$($b$)$ If $A$ is selfadjoint and $E$ is its spectral measure then
\begin{equation*}
\bbc(A)=\lin\zb{E(\sigma)f}{\sigma \text{ bounded, }f\in\hhc}.
\end{equation*}

 $($c$)$ If $A$ is symmetric and any of $\bbc(A)$,  $\aac(A)$, $\eec(A)$ and $\lin\qqc(A)$ is dense in $\hhc$ then $A$ is essentially selfadjoint.
 \end{facts}

 \subsection{Employing $\ccc^{\infty}$-vectors}\label{s1.12.09} 
 Let us try to engage analytic and quasianalytic vectors in deciding for  which $k$'s the operators $A^{(k,i)}$ are (essentially) selfadjoint. This is the first step towards answering the question of unitarity of so called higher order squeeze operators.

Let us collect first some formulae; the calculations are postponed to Appendix.
 \begin{lem}\label{3Oct13-1} For $k=1,2,3\ldots$ and with $j_{0} = n-r-1$ we have
  \begin{equation}\label{3Oct13-1a}
 \big(A^{(k,i)}\big)^{n} e^{(k,i)}_{p}
  = \sum_{r=0}^{n} \prod_{s=1}^{r} \sum_{j_{s}=s-1}^{j_{s-1}+1} \frac{[i+(p-r+1+j_{s})k]!}{[i+(p-r+j_{s})k]!} \sqrt{\frac{[i+(p+n-2r)k]!}{(i+pk)!}} (-\I\E^{\I\theta})^{n-2r} e^{(k,i)}_{p+n-2r},
 \end{equation}
  \begin{align}
\begin{split}\label{3Sep13-4}
\| \big(A^{(k, i)}\big)^n e^{(k, i)}_{p} \|^{2} = \sum_{r=0}^{n} \prod_{s=1}^{r} \sum_{j_{s}=s-1}^{j_{s-1}+1} \left\{\frac{[i+(p-r+1+j_{s})k]!}{[i+(p-r+j_{s})k]!}\right\}^{2} \frac{[i+(p+n-2r)k]!}{(i+pk)!},
\end{split}
\end{align}
\begin{equation}\label{1.29.01}
 \sqrt{\frac{[i + (p+n)k]!}{(i+pk)!}}\Le\| \big(A^{(k, i)}\big)^n e^{(k, i)}_{p} \| \Le \sqrt{2}k^{\frac{n}{2}} \sqrt{\frac{[i + (p+n)k]!}{(i+pk)!}}.
\end{equation}
 \end{lem}

\begin{pro}\label{9Sep13-5} \liczp 1. The series
\begin{equation}\label{9Sep13-6}
\sum_{n=0}^{\infty} \frac{\|(A^{(k, i)})^{n} e^{(k, i)}_{p}\| t^{n}}{n!}
\end{equation}
converges for $k=1$ with infinite radius of convergence (entire vectors) and for  $k=2$ with $t < 1/(2\sqrt{2})$ (analytic vectors). 

\liczp 2. The series \eqref{9Sep13-6} diverges for $k = 3, 4, \ldots$
and furthermore the series 
\begin{equation*}
\sum_{n=0}^{\infty} \| \big(A^{(k, i)}\big)^{n} e^{(k, i)}_{p} \|^{-\ulamek{1}{n}}
\end{equation*}
 converges for $k {\Ge} 3$.
\end{pro}
\begin{proof}
First we will proof the convergence of \eqref{9Sep13-6} for $k=1, 2$. \eqref{1.29.01}
 gives
\begin{equation}\label{9Sep13-7}
\sum_{n=0}^{\infty} \frac{\|\big(A^{(k, i)}\big)^{n} e^{(k, i)}_{p}\| t^{n}}{n!} \Le \frac{\sqrt{2}}{\sqrt{(i+kp)!}} \sum_{n=0}^{\infty} \frac{\sqrt{[i+(p+n)k]!}}{n!} k^{\ulamek{n}{2}} t^{n}. 
\end{equation}
Using d'Alembert's test of convergence for the  series on the right hand side of  \eqref{9Sep13-7} we get part \liczp 1 of the Proposition.

To prove the divergence of \eqref{9Sep13-6} for $k=3, 4, \ldots$ we rewrite \eqref{1.29.01} as
\begin{equation}\label{9Sep13-8}
\sum_{n=0}^{\infty} \frac{\|\big(A^{(k, i)}\big)^{n} e^{(k, i)}_{p}\| t^{n}}{n!} \geq \sum_{n=0}^{\infty} \frac{t^{n}}{n!} \sqrt{\frac{[i+(p+n)k]!}{(i+pk)!}}. 
\end{equation}
Employing  d'Alembert's test  to the right hand side of \eqref{9Sep13-8} we conclude the series in \liczp 2 is divergent. 

Lemma \ref{3Oct13-1} and Stirling's formula give
\begin{align*}\label{9Sep13-2}
&\sum_{n=0}^{\infty} \| \big(A^{(k, i)}\big)^{n} e^{(k, i)}_{p} \|^{-\ulamek{1}{n}} \Le \sum_{n=0}^{\infty} \left\{\frac{[i + (p+n)k]!}{(i+pk)!}\right\}^{-\ulamek{1}{2n}} \Le \sum_{n=0}^{\infty} [(i + (p+n)k)!]^{-\ulamek{1}{2n}} 
\\&
\Le \sum_{n=0}^{\infty} [i + (p+n)k]^{-\ulamek{i + pk + 1/2}{2n} - \ulamek{k}{2}}
 \exp\Big\{\ulamek{i + pk}{2n} + \ulamek{k}{2} - \ulamek{1}{2n[1 + 12i + 12(p+n)k]}\Big\} \\&
\qquad \Le \sum_{n=0}^{\infty} k^{-\ulamek{k}{2}} e^{\ulamek{k}{2}} n^{-\ulamek{k}{2}} = \left(\frac{e}{k}\right)^{\ulamek{k}{2}} \sum_{n=0}^{\infty} n^{-\ulamek{k}{2}}< +\infty,\quad k\Ge3.\qedhere
\end{align*}
\end{proof}
\begin{cor}\label{t1.23.09}
$A_{\xi}^{(k)}$ is essentially selfadjoint for $k=1,2$ while any of $A^{(k,i)}$ $($hence $A_{\xi}$$)$ \underbar{may} not be so if $k\Ge3$ (because \liczp 2 is a necessary, not a sufficient condition for essential slfadjointness).
\end{cor}
The above confirms what is already recognized in this or another way for $k=1,2$; for $k\Ge 3$  it leaves the question  unfastened for the time being.

\section{Essential selfadjointness of the operators $A^{(k,i)}$. The second attempt - definite}\label{s2.22.09a}

	\subsection{The deficiency index approach to essential selfadjointeness} \label{def_{ind}}
The deficiency indices (sometimes called the defect numbers) $n_{+}$ and $n_{-}$ of a symmetric operator $A$ in a Hilbert space $\hhc$ are defined as follows
\begin{equation*}
n_{\pm} = \dim\,\ob{A \pm \I}^{\perp}.
\end{equation*}
It is included in the classical von Neumann theory of selfadjoint extensions of symmetric operators  that $A$ is essentially selfadjoint (that is, its closure is selfadjoint), if and only if 
\begin{equation}\label{1.15.02}
n_{+}=n_{-}=0.
\end{equation}
Furthermore, the main part of the theory ensures the existence of selfadjoint extensions in the same space (that is in the sense described in the second paragraph of  Subsection \ref{basic1}, Section \ref{s2.22.09}) precisely when both deficiency indices are equal.

  \subsection{Towards determining the deficiency indices of $A^{(k,i)}$'s}\label{1s.16.06}
   In order to determine the deficiency indices of $A^{(k,i)}$ take 
  $f\in\hhc^{(i)}$  and check the cardinality of linearly independent $f$'s orthogonal to $\ob{A^{(k,i)} \pm \I}$ for both $\pm\I$, which reads as
  \begin{equation}\label{2.15.02}
  \langle ({A^{(k,i)}} \pm \I)e^{{(k, i)}}_{p}, f \rangle = 0,\quad p=0,1,\ldots.
  \end{equation}
  Notice that, due to the third of \eqref{a}, 
  \begin{equation}\label{1.15.06}
  {A^{(k,i)}}e^{{(k, i)}}_{p}= {-}\I\E^{\I\theta}a_{+}^{k}e^{{(k, i)}}_{p}\; \text{ for }\; p = -k, -k+1, \ldots, -1, {0}.
  \end{equation}
    
  Develop $f$ as $f = \sum_{\alpha}f^{{(k, i)}}_{\alpha} e^{{(k, i)}}_{\alpha}$, $f^{{(k, i)}}_{\alpha} = \langle e^{{(k, i)}}_{\alpha}, f \rangle$  and write according to \eqref{1.18.06} the left hand side of  \eqref{2.15.02} as follows 
  \begin{gather*}
\langle (A^{(k,i)}\pm \I)e^{(k, i)}_{p}, f \rangle = -\I \sum_{\alpha}\nolimits f^{(k, i)}_{\alpha} \langle (\E^{\I\theta}a_{+}^{k} -\E^{-\I\theta} a_{-}^{k} \mp 1)e^{(k, i)}_{p}, e^{(k, i)}_{\alpha} \rangle\\
 -\sqrt{\ulamek{(i+pk)!}{[i+(p-1)k]!}} \E^{-\I\theta} \langle e^{(k, i)}_{p-1}, e^{(k, i)}_{\alpha}\rangle \mp 
  \langle e^{(k, i)}_{p}, e^{(k, i)}_{\alpha}\rangle\big\} 
\\
= -\I \big\{\sum_{\alpha}\nolimits \E^{\I\theta}f^{(k, i)}_{\alpha} \sqrt{\ulamek{[i+(p+1)k]!}{(i+pk)!}} \delta_{p+1, \alpha} -\E^{-\I\theta} \sum_{\alpha}\nolimits f^{(k, i)}_{\alpha} \sqrt{\ulamek{(i+pk)!}{[i+(p-1)k]!}} \delta_{p-1, \alpha}  \mp\sum_{\alpha}\nolimits f^{(k, i)}_{\alpha} \delta_{p, \alpha}\big\} 
  \\ 
= -\I \big\{\sqrt{\ulamek{[i+(p+1)k]!}{(i+pk)!}} \E^{\I\theta} f^{(k, i)}_{p + 1} - \sqrt{\ulamek{(i+pk)!}{[i+(p-1)k]!}} \E^{-\I\theta} f^{(k, i)}_{p-1} \mp f^{(k, i)}_{p}\big\}.
\end{gather*}

Now \eqref{2.15.02}    now takes the form
\begin{equation}\label{08Feb2013-3}
\sqrt{\ulamek{[i+(p+1)k]!}{(i+pk)!}} \E^{\I\theta} f^{(k, i)}_{p + 1} - \sqrt{\ulamek{(i+pk)!}{[i+(p-1)k]!}} \E^{-\I\theta} f^{(k, i)}_{p-1} \mp f^{(k, i)}_{p} = 0,\;\; p=0,1,\ldots, 
\end{equation}  
with 
\begin{equation}\label{2.15.06}
f^{(k, i)}_{-q}\okr f^{(k, i)}_{-q+1}\okr \ldots = f^{(k, i)}_{{-1}}\okr 0
\end{equation}
which is in accordance with \eqref{1.15.06}. 

Let us treat the cases $k=1,2$ and $k\Ge 3$ separately.

\subsection{\underbar{The ups:} the cases $k=1$ and $k=2$ }\label{s1.05.017}
It is well known that for the measure orthogonalizing  polynomials $\pi_{n}$, $n=0,1\ldots$, to be unique (or, in other  words, the corresponding moment problem to be determinate) it is necessary and sufficient
\begin{equation}\label{3.15.02}
\sum_{p=0}^{\infty}|\pi_{p}(z)|^{2}=+\infty
\end{equation}
for any $z$ with $\IM z\neq 0$;
cf. \cite[Theorem 3]{simon}. This means  the would-be Fourier coefficients $f_{n}$ are not in $\ell^{2}$ which leaves the hypothetical  $f$ out of the space $\hhc$ and is a counterpart of Proposition \ref{9Sep13-5} part  \liczp 1; both Hermite and Meixner-Pollaczek polynomials are determinate.

\subsubsection{Case $k=1$}\label{exa1}
Here $i=0$ is the only possibility and the formulae \eqref{08Feb2013-3} and \eqref{2.15.06}, after setting $g_{p}\okr\E^{-\I\!p\theta}\!f^{(1, 0)}_{p}$, take the form
\begin{equation}\label{08Feb2013-4}
\sqrt{p+1}g_{p+1} - \sqrt{p}g_{p-1} \mp g_{p} = 0,\quad p = 0,1\ldots,\; g_{-1}=0.
\end{equation}
If {$g_{0}=0$}  we get immediately that  the only solution of  \eqref{08Feb2013-4} is {$g_{n}=0$} for all $n=0,1,\ldots$ and both $\mp$. If not, then supposing {$g_{0}=1$} we can proceed as follows.

Normalizing the Hermite polynomials as $h_{p}(x) =  \I^{p} H_{p}(x)/\sqrt{2^{p} p!}$ from 
the standard recurrence relation 
\begin{equation*}
H_{p+1}(x) + 2p H_{p-1}(x) - 2x H_{p}(x) = 0,
\end{equation*}
one gets
\begin{equation}\label{08Feb2013-6}
\sqrt{p+1} h_{p+1}(x) - \sqrt{p}h_{p-1}(x) - x\!\sqrt{2} x h_{p}(x) = 0.
\end{equation}
Comparing \eqref{08Feb2013-4} and \eqref{08Feb2013-6} and taking into account that $g_{0}=h_{0}=1$  and that the Hermite polynomials are the only solutions of \eqref{08Feb2013-6} we infer that
\begin{equation*}
g_{p} = h_{p}\Big(\!\pm \ulamek\I{\sqrt 2}\Big).
\end{equation*}
Consequently, due to \eqref{3.15.02} the solution within $\hhc^{(1)}$ is  $g_{p}=0$, $p=0,1\ldots$

\subsubsection{Case $k=2$}\label{exa2}
Considering two parallel cases $i=0$ and $i=1$ we have to take into account Corollary \ref{t1.7.07} which results in splitting \eqref{2.15.02} . Thus the formula \eqref{08Feb2013-3} also splits in two, $i=0,1$, 
\begin{equation}\label{08Feb2013-10}
\sqrt{(2p+i+2)(2p+i+1)}\E^{\I\theta} f^{(2, i)}_{p+1} - \sqrt{(2p+i)(2p+i-1)}\E^{-\I\theta} f^{(2, i)}_{p-1} \mp f^{(2, i)}_{p} = 0
\end{equation}
with $p=0, 1, \ldots$ and $f^{(2, i)}_{-1} = 0$. The conditions $f^{(2, i)}_{0} = 0$ imply $f^{(2, i)}_{p} = 0$. Henceforth we take $f^{(2, i)}_{0} = 1$. 

For $c^{(i)}_{p} = \I^{p} \E^{\I\! p\theta} f^{(2, i)}_{p}$, $i=0, 1$, \eqref{08Feb2013-10} reads
\begin{equation}\label{19Jun13-1}
\sqrt{(p+\ulamek{i+2}{2})(p + \ulamek{i+1}{2})} c^{(i)}_{p+1} \mp \ulamek{\I}{2} c^{(i)}_{p} + \sqrt{(p+\ulamek{i}{2})(p + \ulamek{i-1}{2})} c^{(i)}_{p-1} = 0.
\end{equation}
Consider the Meixner-Pollaczek polynomials $P^{(\lambda)}_{n}(\;\cdot\;;\ulamek \pi 2)$, $n=0,1,\ldots$  and normalize them according to formula  (9.7.2) in \cite[p. 213]{koe}
\begin{equation}\label{1.19.02}
p^{(\lambda)}_{n}\okr{\sqrt{\frac{2^{2\lambda}n!}{2\pi \Gamma(n + 2\lambda)}}}\,P^{(\lambda)}_{n}(\,\cdot\,;\ulamek \pi 2)
\end{equation}
which sends the recurrence relation
$$
(n+1)P^{(\lambda)}_{n+1}(x;\ulamek \pi 2) - 2 x P^{(\lambda)}_{n}(x;\ulamek \pi 2) + (n+2\lambda -1)P^{(\lambda)}_{n-1}(x;\ulamek \pi 2)=0
$$
into
\begin{equation}\label{2.19.02}
\sqrt{(n+1)(n+2\lambda)}p^{(\lambda)}_{n+1}(x) - 2xp^{(\lambda)}_{n}(x) + \sqrt{n(n+2\lambda - 1)}p^{(\lambda)}_{n-1}(x)=0.
\end{equation}
Comparing \eqref{19Jun13-1} and \eqref{2.19.02} for appropriate $i$ we get following two couples: $i=0$ corresponding to   $\lambda = \ulamek{1}{4}$ and $i=1$ corresponding to $\lambda = \ulamek{3}{4}$. Moreover we have
\begin{equation*}
c^{(0)}_{n}=p^{(\tiny\frac 14)}_{n}{\Big(\!\pm\ulamek \I 4\Big)} \quad \text{and} \quad c^{(1)}_{n}=p^{(\tiny\frac 34)}_{n}{\Big(\!\pm\ulamek \I 4\Big)}.
\end{equation*}
For the same reason as above the series
\begin{equation*}
\sum_{n=0}^{\infty}|c^{(i)}_{n}|^{2}=\sum_{n=0}^{\infty}\Big|p^{(\lambda)}_{n}\Big(\!\pm\ulamek \I 4\Big)\Big|^{2}, 
\end{equation*}
are divergent for both $\pm$. Because they are a subseries of 
\begin{equation*}
\sum_{n=0}^{\infty}|f^{(2, i)}_{n}|^{2}, \quad i=0,1
\end{equation*}
the latter are divergent as well. The argument goes like in the case $k=1$  before.

\subsection{\underbar{Downs}: the case $k\Ge 3$}\label{28.08}

The recurrence \eqref{08Feb2013-3}, after fixing $i=0,1,\ldots,k-1$ and introducing $d^{(k, i)\, \pm}_{p} = \E^{\I\!p\theta}f^{(k, i)}_{p}$, takes the form
\begin{equation}\label{29Feb2013-3}
d^{(k,i)\, \pm}_{p+1} = \frac{(i+pk)!}{\sqrt{(i+pk-k)! (i+pk+k)!}}\, d^{(k, i)\, \pm}_{p-1} \pm \sqrt{\frac{(i+pk)!}{(i+pk+k)!}}\, d^{(k, i)\, \pm}_{p}
\end{equation}  
with \eqref{2.15.06} turning to
\begin{equation}\label{2.16.04}
d^{(k, i)\, \pm}_{-1}=0, \quad i=0,1, \ldots, k-1.
\end{equation}
\begin{rem}\label{rem1}
With \eqref{2.16.04} the zero sequence is the only solution of \eqref{29Feb2013-3} for each  initial conditions $d^{(k, i)+}_{0}=0$ or $d^{(k, i)-}_{0}=0$.
This implies that \eqref{29Feb2013-3} has at most one solution for each  of the cases $+$ and $-$.
 \end{rem}

\begin{rem}\label{2.1}
If $d^{(k, i)+}_{0}>0$ then so are all the other entries of the sequence $(d^{(k, i)+}_{p})_{p}$. This can be inspected by induction applied to \eqref{29Feb2013-3}.
\end{rem}

\begin{rem} \label{rem3}
If $d^{(k, i)+}_{0} \neq 0$ then $ d^{(k, i) +}_{p} \neq 0$ for all $p$. 

Suppose the contrary, there exists $p$ such that $d^{(k, i) +}_{p+1} = 0$ and let it be the smallest such. By \eqref{29Feb2013-3} we have
\begin{equation*}
\sqrt{\frac{(i+kr)!}{(i + kr -k)!}}\, d^{(k, i) +}_{r-1} + d^{(k, i) +}_{r} = 0.
\end{equation*}
Consequently $d^{(k, i) +}_{p-1} $ and $d^{(k, i) +}_{p} $ are of different sings which contradicts Remark \ref{2.1}.
\end{rem}

\begin{pro}\label{pro2} 
$d^{(k, i) -}_{p} = (-1)^{p} d^{(k, i) +}_{p}$ for all $p$. 
\end{pro}
 
\begin{proof}
From \eqref{29Feb2013-3} for $d^{(k, i)\, -}_{p}$ we get
\begin{align*}
(-1)^{p+1} d^{(k, i)\, +}_{p+1} &= \frac{(i+pk)!}{\sqrt{(i+pk-k)! (i+pk+k)!}}\, (-1)^{p-1} (-1)^{2}d^{(k, i)\, +}_{p-1} - \sqrt{\frac{(i+pk)!}{(i+pk+k)!}} (-1)^{p} d^{(k, i)\, +}_{p} 
\end{align*}
which shortens to
$$
(-1)^{p+1} d^{(k, i)\, +}_{p+1} = \frac{(i+pk)!}{\sqrt{(i+pk-k)! (i+pk+k)!}}\, (-1)^{p-1} d^{(k, i)\, +}_{p-1} - \sqrt{\frac{(i+pk)!}{(i+pk+k)!}} (-1)^{p} d^{(k, i)\, +}_{p}
$$
Comparing this with \eqref{29Feb2013-3} for ``-'' and using the uniqueness in Remark \ref{rem1} we get the conclusion.
\end{proof}

\begin{rem}\label{1} Proposition \ref{pro2} implies $|d^{(k, i) -}_{p}| = |d^{(k, i) +}_{p}| = d^{(k, i) +}_{p}$ for all $p$. Denote this common number shortly by $d^{(k, i)}_{p}$. Henceforth, we can examine exclusively the  equation
\begin{equation}\label{eq3}
d^{(k, i)}_{p+1} = \frac{(i+pk)!}{\sqrt{(i+pk-k)! (i+pk+k)!}}\, d^{(k, i)}_{p-1} + \sqrt{\frac{(i+pk)!}{(i+pk+k)!}}\, d^{(k, i)}_{p}
\end{equation}
for $d^{(k, i)}_{p}$'s.
\end{rem}

\begin{pro}{\label{rem26.03.13-1} With notation $\alpha^{(k)}_{p} \okr \sqrt{\ulamek{(i+pk)!}{(i+pk+k)!}}$
\begin{equation}\label{1.10.04}
d^{(k, i)}_{p+1} - d^{(k, i)}_{p-1} < \alpha^{(k)}_{p} \alpha^{(k)}_{p-1} \ldots \alpha^{(k)}_{1} \alpha^{(k)}_{0} \big(d^{(k, i)}_{2} - d^{(k, i)}_{0}\big)
\end{equation}}
\end{pro}
\begin{proof} 
Prove first
\begin{equation*}
d^{(k, i)}_{p+1} - d^{(k, i)}_{p-1} < \alpha^{(k)}_{p} \left(d^{(k, i)}_{p} - d^{(k, i)}_{p-2}\right)
\end{equation*}

Using \eqref{eq3}, \eqref{1.4.2} (notice \eqref{1.4.2} \underbar{requires} $k\Ge3$) and \eqref{2.4.2} we have
\begin{align*}
d^{(k, i)}_{p} &= \frac{(i+pk-k)!}{\sqrt{(i+pk)! (i+pk-2k)!}} d^{(k, i)}_{p-2} + \sqrt{\frac{(i+pk-k)!}{(i+kp)!}} d^{(k, i)}_{p-1} \\[0.5\baselineskip]
&< \frac{(i+pk)!}{\sqrt{(i+pk-k)! (i+pk+k)!}} d^{(k, i)}_{p-2} + \left(1 - \frac{(i+pk)!}{\sqrt{(i+pk-k)! (i+pk+k)!}}\right) d^{(k, i)}_{p-1}.
\end{align*}
That gives
\begin{equation*}
\frac{(i+pk)!}{\sqrt{(i+pk-k)! (i+pk+k)!}} d^{(k, i)}_{p-1} < \frac{(i+pk)!}{\sqrt{(i+pk-k)! (i+pk+k)!}} d^{(k, i)}_{p-2} + d^{(k, i)}_{p-1} - d^{(k, i)}_{p}.
\end{equation*}
Inserting into \eqref{eq3}, we get
\begin{align*}
d^{(k, i)}_{p+1} - d^{(k, i)}_{p-1} &< \frac{(i+pk)!}{\sqrt{(i+pk-k)! (i+pk+k)!}} d^{(k, i)}_{p-2} + \left(\sqrt{\frac{(i+pk)!}{(i + pk + k)!}} - 1\right) d^{(k, i)}_{p} \\
& < \left(\sqrt{\frac{(i+pk)!}{(i + pk + k)!}} - 1\right) \left(d^{(k, i)}_{p} - d^{(k, i)}_{p-2}\right) < \sqrt{\frac{(i+pk)!}{(i + pk + k)!}} \left(d^{(k, i)}_{p} - d^{(k, i)}_{p-2}\right).
\end{align*}
Now the induction argument makes \eqref{1.10.04}.
\end{proof}

\begin{cor}\label{t2.10.02}
The sequence $(d^{(k, i)}_{p})_{p}$ is convergent.
\end{cor}
\begin{proof}
It is clear that 
\begin{equation}\label{25Apr13-2}
\sum_{r=p}^{\infty}\alpha^{(k)}_{r}\alpha^{(k)}_{r-1}\cdots\alpha^{(k)}_{0} = \sum_{r=p}^{\infty} \sqrt{\frac{i!}{(kr+k+i)!}}<+\infty. 
\end{equation}
Notice that
\begin{align}\label{1bis.16.04}
\begin{split}
|d^{(k, i)}_{p+m}-d^{(k, i)}_{p}&|\Le|d^{(k, i)}_{p+m}-d^{(k, i)}_{p+m-2}|+|d^{(k, i)}_{p+m-2}-d^{(k, i)}_{p+m-4}|+\cdots+|d^{(k, i)}_{p+2}-d^{(k, i)}_{p}|\\&
\Le\big(d^{(k, i)}_{2}-d^{(k, i)}_{0}\big)\sum_{r=p}^{m+p}\alpha^{(k)}_{r}\alpha^{(k)}_{r-1}\cdots\alpha^{(k)}_{0}.
\end{split}
\end{align}
Because RHS is equal to $\big(d^{(k, i)}_{2}-d^{(k, i)}_{0}\big)$ multiplied a Cauchy fragment  of a convergent series \eqref{25Apr13-2}
 LHS tends to $0$ which says $(d^{(k, i)}_{p})_{p}$ is a Cauchy sequence, hence it is convergent.
\end{proof}
\begin{thm}\label{t1.10.02}
None of the operators $A^{(k,i)}$, $k\Ge3$ and  $i=0,\ldots,k-1$, is essentially selfadjoint.
\end{thm}
\begin{proof}
As already experienced it is enough to show that the series $\sum_{p=0}^{\infty}(d^{(k, i)}_{p}){}^{2}$ is convergent. We already know, Corollary \ref{t2.10.02}, the sequence  $(d^{(k, i)}_{p})_{p}$ is convergent.
Then either
\begin{equation}\label{4.16.04}
\lim_{p\to+\infty}d^{(k, i)}_{p}\neq0.
\end{equation}
or
\begin{equation}\label{3.16.04}
\lim_{p\to+\infty}d^{(k, i)}_{p}=0
\end{equation}

Let us go on with \eqref{1bis.16.04} as follows\begin{align}\label{1ter.16.04}
\begin{split}
|d^{(k, i)}_{p+m}-d^{(k, i)}_{p}|&
\Le\big(d^{(k, i)}_{2}-d^{(k, i)}_{0}\big)\sum_{r=p}^{m+p}\alpha^{(k)}_{r}\alpha^{(k)}_{r-1}\cdots\alpha^{(k)}_{0}
={\rm const}\,\alpha^{(k)}_{p}\cdots\alpha^{(k)}_{0}
\\&\times\sum_{r=p+1}^{m+p}\alpha^{(k)}_{r}\alpha^{(k)}_{r-1}\cdots\alpha^{(k)}_{p+1}
\Le{\rm const}\,\alpha^{(k)}_{p}\cdots\alpha^{(k)}_{0}\sum_{r=p+1}^{\infty}\alpha^{(k)}_{r}\alpha^{(k)}_{r-1}\cdots\alpha^{(k)}_{p+1}\\&\stackrel{\eqref{25Apr13-2}}{=}{\rm const}\,\sum_{r=p+1}^{\infty} \sqrt{\frac{i!}{(kr+k+i)!}}
\Le{\rm const}\,\sum_{r=0}^{\infty} \sqrt{\frac{i!}{(kr+k+i)!}}. 
\end{split}
\end{align}

Suppose \eqref{4.16.04} holds. Then for the sequence
\begin{equation*}
x^{(k, i)}_{p} \okr \frac{d^{(k, i)}_{p}}{d^{(k, i)}_{p+1}}.
\end{equation*}
we get immediately $x^{(k, i)}_{p}\to1$ as $p$ goes to $+\infty$.

From \eqref{1.10.04} we have $d^{(k, i)}_{p} < d^{(k, i)}_{p -2}$ and $x^{(k, i)}_{p} < x^{(k, i)}_{p -2} x^{(k, i)}_{p-1} x^{(k, i)}_{p}$. Because
\begin{equation*}
\lim_{p\to+\infty} \sqrt{\frac{(i+pk-k)!}{(i+pk)!}} = 0 \quad \text{and} \quad \lim_{p\to+\infty} \frac{\sqrt{(i+pk-k)! (i+pk+k)!}}{(i+pk)!} = 1.
\end{equation*}
we get
\begin{align*}
\lim_{p\to+\infty} p \left(x^{(k, i)}_{p} - 1\right) &< \lim_{p\to+\infty} p \left(x^{(k, i)}_{p -2} x^{(k, i)}_{p-1} x^{(k, i)}_{p} - 1\right) = \lim_{p\to+\infty} p \left(\frac{\sqrt{(i+pk-k)! (i+pk+k)!}}{(i+pk)!}\, x^{(k, i)}_{p-2} \right.\nonumber \\
&\left.- \sqrt{\frac{(i+pk-k)!}{(i+pk)!}}\, x^{(k, i)}_{p-2} x^{(k, i)}_{p} - 1\right) = \frac{k}{2}
\end{align*}
and consequently
\begin{equation*}
\lim_{p\to+\infty} p \left[\big(x^{(k, i)}_{p}\big)^{2} - 1\right] = \lim_{p\to+\infty} p \left(x^{(k, i)}_{p} - 1\right) \lim_{p\to+\infty} \left(x^{(k, i)}_{p} +1\right) < 2 \frac{k}{2} = k.
\end{equation*}
Due to Raabe's criterion we have 
\begin{equation}\label{7}
\sum_{p=0}^{\infty}d^{(k,i)}_{p}
\end{equation}
is convergent. This excludes the case \eqref{4.16.04} to hold.

If \eqref{3.16.04} happens, then passing in \eqref{1ter.16.04} with $m$ to $+\infty$ we get
\begin{align*}
|d^{(k, i)}_{p}|&\Le{\rm const}\,\sum_{r=0}^{\infty} \sqrt{\frac{i!}{(kr+k+i)!}}\end{align*}
Because in this case for $p$ sufficiently large $d^{(k, i)}_{p}{}^{2} \Le d^{(k, i)}_{p}$ we have
\begin{equation*}
\sum_{p=0}^{\infty}(d^{(k, i)}_{p}{})^{2}\Le {\rm{const}'}\sum_{p=0}^{\infty}d^{(k, i)}_{p}\Le{\rm const}''\sum_{p=0}^{\infty}\alpha^{(k)}_{r}\alpha^{(k)}_{r-1}\cdots\alpha^{(k)}_{0}\sum_{r=0}^{\infty} \sqrt{\frac{i!}{(kr+k+i)!}}.
\end{equation*}
 Therefore  convergence of the series \eqref{7} has been proved.
\end{proof}

\section{Generating $\exp\{\I \!tA^{(k)}_{\xi}\}$. Squeeze operators of any order?}\label{s2.7.07}
\subsection{The groundwork thought over}\label{s3.7.07}
Suppose we are given a selfadjoint operator $B$, if $E$ stands for its spectral measure then the spectral integral 
\begin{equation}\label{1.8.07}
\int_{\rrb}\E^{\I\!tx}E(\D x)
\end{equation}
(understood as usually in the weak or strong operator topology) gives rise according to the rules of functional calculus   to a one parameter ($t\in\rrb$) family of unitary operators which is customarily denoted as $\E^{\I\!tB}$. Due to the continuity property of spectral integral it is strongly continuous in $t$.  If $B$ is essentially selfadjoint then its closure $\sbar B$ is selfadjoint so one can think of $\E^{\I\!t\sbar B}$.

On the other hand one has a definition: a family $\{U(t)\}_{t\in\rrb}$ of unitary operators in $\hhc$ is said to be a {\em strongly continuous one-parameter unitary group} if 
\begin{enumerate}
\item[(a)] $U(s+t)=U(s)U(t)$, $s,t\in\rrb$;
\item[(b)] $\lim_{h\to 0} \big(U(t+h)-U(t)\big)f=f$ for $f\in\hhc$ and $t\in\rrb$.
\end{enumerate}
It is clear that the  exponential family $(\E^{\I\!tB})_{t\in\rrb}$ just defined is a strongly continuous one-parameter unitary group.
The celebrated Stone theorem shows the way back: every strongly continuous one-parameter unitary group is of the form $(\E^{\I\!tB})_{t\in\rrb}$ with a uniquely determined selfadjoint operator $B$; it establishes a bijection between $(\E^{\I\!tB})_{t\in\rrb}$ and $(U(t))_{t\in\rrb}$ making them to be replaceable. In conclusion, the spectral integral definition \eqref{1.8.07} of $(\E^{\I\!tB})_{t\in\rrb}$ is the primary way to defining the unitary group in question and this is made possible at least.

The operator $B$ is pretty often called the (infinitesimal) {\em generator} of the group $(U(t))_{t\in\rrb}$ and is defined by 
\begin{align*}
\dz B\okr\zb{f\in\hhc}{\poch{}t{}U(t)|_{t=0}f\okr\lim_{h\to0}h^{-1}\big(U(h)-I\big)f\text{ exists}},\quad
\I\! Bf\okr \poch{}t{}U(t)|_{t=0}f,\quad f\in\dz B.
\end{align*}
As a kind of  extras attached to Stone's theorem we have for  $t\in\rrb$
\begin{gather} \label{3.8.07}
U(t)\dz B\subset \dz B,\quad
\poch{}t{}U(t)f=\I\!U(t)Bf=\I\!BU(t)f \text{ for }f\in\dz B.
\end{gather}
\begin{rem}\label{t1.11.01}
The repeated use of the second of \eqref{3.8.07} leads to
\begin{equation}
\poch{}t{n}U(t)f=\I^{n}U(t)B^{n}f=B^{n}U(t)f \text{ for }f\in\bigcap_{i=0}^{n}\dz {B^{i}}.\label{3a.8.07}
\end{equation}
Therefore the question is for which $f$'s the Taylor series (the very left expression)
\begin{equation*}
U(t)f=\sum_{n=0}^{\infty}\frac{t^{n}}{n!}\poch{}t{n}U(t)\res{t=0}f=\sum_{n=0}^{\infty}\frac{(\I\!t)^{n}}{n!}U(t)B^{n}f=\sum_{n=0}^{\infty}\frac{(\I\!t)^{n}}{n!}B^{n}f.
\end{equation*}
converges and how.
\end{rem}

\subsection{More on the role of analytic vectors}\label{s2.8.07} 
The spectral integral \eqref{1.8.07} allows to determine the unitary group once the spectral measure of its generator is known. A practical question is if one can avoid the spectral representation trying to promote suggestively a kind of Taylor series expansion by means of $\ccc^{\infty}$ vectors. More precisely, starting with an essentially selfadjoint operator $A$ the question is for which $f$'s the definition
\begin{equation*}
U(t)f\okr\sum_{n=0}^{\infty}\frac{(\I\!t)^{n}}{n!}A^{n}f
\end{equation*}
makes sense.
This has to be handled with some caution. An insight into the proof of Lemma 5.1 in \cite{nel} shows how to make this construction possible in the case when the set of analytic vectors $\aac{(A)}$ of a symmetric operator $A$ is dense.

 The construction in \cite[Lemma 5.1]{nel} is local and can be reiterated  resulting in the desired group. In the case  when the set of entire vectors $\eec(A)$ is dense the construction can be made smoother, giving at once the group $(U(t))_{t\in\rrb}$ the operator $\sbar A$ generates.
 Therefore $\eec(A)$ is a subpace of $\hhc$ for which one
can certainly replace integral with summation in the middle equality of
\begin{equation}\label{1.9.07}
\E^{\I\!t\sbar A}f\stackrel{\eqref{1.8.07}}=\int_{\rrb}\;\sum^{\infty}_{n=0}\frac{(\I\!t)^{n}}{n!}x^{n}E(\D x)f
=\sum^{\infty}_{n=0}\int_{\rrb}\frac{(\I\!t)^{n}}{n!}x^{n}E(\D x)f
=\sum^{\infty}_{n=0}\frac{(\I\!t)^{n}}{n!}A^{n}f,\; f\in\eec(A);
\end{equation}
as a matter of fact the first equality holds for all $f\in\hhc$.

\begin{rem}\label{t1.19.01}
The role of $\ccc^{\infty}$ vectors in determining essential selfadjointness is described in some details in Fact \ref{t3.30.11}. Though selfadjoint operators themselves have enough $\ccc^{\infty}$ vectors of any kind appearing there, an essentially selfadjoint ones (in particular the candidate for) may not have  even quasianalytic vectors, they simply may not fit in with the domain of an operator which {\it a priori} is not closed. This makes the unseen at a first glance difference we want to put strong emphasis on. An acute awareness of this fact helps to monitor the situation we are in.
\end{rem}

	\subsection{The displacement and squeeze operators}\label{s3.22.09}
	Corollary \ref{t1.23.09} or, alternatively, the results of Subsection \ref{s1.05.017} lead directly to what the majority of mathematical physicist accept as granted (notice that what  customarily emerges in the definitions is the  complex parameter $\alpha=-\I\xi$ ).
	\begin{thm}\label{t1.22.09}
	The displacement $\exp\{\I tA_{\xi}^{(1)}\}$ and squeeze $\exp\{\I tA_{\xi}^{(2)}\}$ operators form a group of unitaries for $t\in\rrb$.
	\end{thm}

	Keeping up with the notations of Subsection \ref{s1.10.11} the important information launched in Corollary \ref{t1.7.07} can be encapsulated now. It sheds more light on how the squeeze operator behaves.
	\begin{cor}\label{t1.10.11}
	The Hilbert space $\hhc$ decomposes as $\hhc=\hhc^{(2,0)}\bigoplus\hhc^{(2,1)}$ with\, $\hhc^{(2,0)}=\clolin\zb{e_{2n}}{n=0,1,\ldots}$ and $\hhc^{(2,1)}=\clolin\zb{e_{2n+1}}{n=0,1,\ldots}$. It forces  the squeeze operators $\exp\{\I tA_{\xi}^{(2)}\}$ to decompose accordantly as
	\begin{equation*}
	\exp\{\I tA_{\xi}^{(2)}\}=\exp\{\I t|\xi|A^{(2,0)}\}\bigoplus\exp\{\I t|\xi|A^{(2,1)}\}
	\end{equation*}
	where $A^{(2,i)}$, $i=0,1$, are Jacobi operators acting as stated by \eqref{1.18.06} in $\hhc^{(2,i)}$ respectively. 
	
	Moreover, the operators $\exp\{\I t|\xi|A^{(2,i)}\}$, $i=0,1$, can be retrieved from \eqref{1.18.06} on the linear spaces $\ddc^{(2,i)}$ defined by \eqref{1.16.06} as they are composed of analytic vectors of the operators $A^{(2,i)}$.
	\end{cor}

\subsection{What happens if $B$ is not essentially selfadjoint - further developments}
Due to Na\u{\i}mark a selfadjoint extension of a symmetric operator  $A$ always exists (cf. \cite[Proposition 3.7]{kon}) if one allows it to be in a larger space, say $\kkc$, isometrically including $\hhc$. On the other hand, if $A$ has  equal deficiency indices, the von Neumann theory provides with  a plenty of selfadjoint extensions still within $\hhc$. Even if $A$ is a Jacobi operator having deficiency indices $(1,1)$, Na\u{\i}mark extensions  are at least as much compelling  as von Neumann ones, look at \cite{darek} for a stimulating example and its analytic background, and some whereabouts at \cite{lms}.

Pick  $B$ be either von Neumann's or Na\u{\i}mark's extension of $A$. Then $(\E^{\I\!tB})_{t\in\rrb}$ is well defined as described above; denote the group alternatively by  $(U^{(B)}(t))_{t\in\rrb}$ stressing on its dependence  on the choice of a selfadjoint extension of $A$.

Passing to the operator $A$ with \underbar{invariant} domain, that is $A\dz A\subset \dz A$, which is our case we can still get something interesting. Because $U^{(B)}(t)Bf=BU^{(B)}(t)f$  for  $f\in\dz B$ (the second part of \eqref{3.8.07}) and because $A\subset B$ we get from \eqref{3a.8.07}

\begin{gather}\label{3a.8.07}
\poch{}t{n}U^{(B)}(t)f=\I^{n}U^{(B)}(t)A^{n}f=\I^{n}B^{n}U^{(B)}(t)f \text{ for $n=0,1\ldots$ and }f\in\dz A
\end{gather}
regardless of the extension $B$.

 Despite the fact that for $k\Ge3$ squeeze operators $\exp\{\I tA_{\xi}^{(k)}\}$
do not exist the situation is not completely hopeless. From the above we get a recipe which  can be read as follows: taking $A=A^{(k,i)}$ for any $i$ with $k$ fixed we get a selfadjoint extension $B^{(k,i)}$ of $A^{(k,i)}$ in some $\kkc^{(k,i)}$ such that
\begin{equation*}
\poch{}t{n}U^{(B^{(k,i)})}(t)f=(\I)^{n}U^{(B^{(k,i)})}(t)(A^{(k,i)})^{n}f=(\I)^{n}(B^{(k,i)})^{n}U^{(B^{(k,i)})}(t)f \; \,\text{for } n=0,1\ldots \text{ and }f\in\dz {A^{(k,i)}}.
\end{equation*}
Summing up the above we come to the  operator $B^{(k)}\okr\bigoplus_{i=0}^{k-1}B^{(k,i)}$, selfadjoint in the space $\kkc^{(k)}\okr\bigoplus_{i=0}^{k-1}\kkc^{(k,i)}$, such that $B^{(k)}$ extends $A^{(k)}\okr\bigoplus_{i=0}^{k-1}A^{(k,i)}$ and
\begin{align*}
\poch{}t{n}U^{(B^{(k)})}(t)f=(\I)^{n}U^{(B^{(k)})}(t)(A^{(k)})^{n}f=(\I)^{n}(B^{(k)})^{n}U^{(B^{(k)})}(t)f\;\,\text{ for $n=0,1\ldots$ and }f\in\dz {A^{(k)}}.
\end{align*}
This opens a lot of possibilities which {we intend to explore  in our future research.}

\section{Back to $k=1$ and $k=2$. Models}\label{aaaa}

Because for $k=1,2$ the operator $A^{(k)}_{\xi}$ is essentially selfadjoint, $(\exp[\I tA^{(k)}_{\xi}])_{t\in\rrb}$ is a group of unitary operators (cf. Theorem \ref{t1.22.09}).

With \framebox{$z=\I t\xi$} we have that the {\em displacement} operator
\begin{equation*}
D(z)=D(t,\xi)\okr\exp[\I t A_{\xi}^{(1)}]=\exp[za_{+}-z^{*}a_{-}],\quad t\in\rrb,\;\xi\in\ccb
\end{equation*}
is unitary   and $D(z^{*})=D(z)^{*}=D(-z)=D(z)^{-1}$; moreover, fixing $\xi\in\ccb$ we have $D(t,\xi)$ to be a group as $t\in\rrb$. The same refers to the {\em squeeze} operator
\begin{equation*}
S(z)=S(t,\xi)\okr\exp[\I t A_{\xi}^{(2)}]=\exp[za_{+}^{2}-z^{*}a_{-}^{2}],\quad t\in\rrb,\;\xi\in\ccb.
\end{equation*}

\subsection{$k=1$; the displacement operator}
\subsubsection{Reviving the models}\label{s1.1.02}
Because the splitting Corollary \ref{t1.7.07}  is not present when $k=1$ ($A^{(1)}_{\xi}=A^{(1,0)}_{\xi}$) the way of proposing  notable expression for the displacement operator can be done just by ``exponentiation'' in the corresponding function space, so to speak. In particular we have at our disposal the following models
\begin{enumerate}
\item[(a)] the $\llc^{2}(\rrb)$ representation (``configuration space'');
\item[(b)] the Segal-Bargmann representation;
\item[(c)] discrete representation by which we mean a one parameter family of harmonic oscillators acting on Charlier sequences considered in $\ell^{2}$, cf. \cite{yet};
\item[(d)] the one parameter family of holomorphic oscillators as done in \cite{anal} (see also \cite{gaz}), which interpolates the models (a) and (b).
\end{enumerate}

Let us take a brisk look at the case (b) as the most analytic one. The orthonormal basis in the Segal-Barmann space is $\big(\ulamek{z^{n}}{\sqrt{n!}}\big)_{n=0}^{\infty}$. Roughly here $D(t,\xi)=\exp[t(\xi z-\xi^{*}\partial_{z})]$ and the unitary equivalence between $\hhc$ and the Segal-Bargmann space is established by
\begin{equation*}
e_{n}\mapsto \ulamek{z^{n}}{\sqrt{n!}},\quad n
=0,1,\ldots,
\end{equation*}
and causes $D(t,\xi)$ to act as a Jacobi operator. Because $\ddc=\lin\big(\ulamek{z^{n}}{\sqrt{n!}}\big)_{n=0}^{\infty}$ is the set of entire vectors for $A^{(1)}_{\xi}$ formula \eqref{1.9.07} applies
\begin{equation*}
D(t,\xi)f=\sum_{n=0}^{\infty}\frac{(\I t)^{n}}{n!}(\xi z-\xi^{*}\partial_{z})^{n}f,\quad f\in\lin\big(\ulamek{z^{n}}{\sqrt{n!}}\big)_{n=0}^{\infty}.
\end{equation*}
In the case of abstract Hilbert space the formula \eqref{3Oct13-1a} (which simplifies substantially as $k=1$ and $i=0$) combined with the the formula \eqref{1.9.07} establishes a series representation of the displacement operator.

\subsection{$k=2$; the squeeze operator}\label{s1.07.11}
The squeeze operator is defined as
\begin{equation*}
S(t,\xi)\okr\exp[\I t A_{\xi}^{(2)}]=\exp[t(\xi a_{+}^{2}-\xi^{*}a_{-}^{2})],\quad t\in\rrb,\;\xi\in\ccb
\end{equation*}
 Benefitting from Corollary \ref{t1.7.07}  the orthogonal splitting\,\footnote{\; It corresponds somehow to what is in \cite{man} where odd and even coherent states are considered.} $\hhc=\hhc^{(2,0)}\bigoplus\hhc^{(2,1)}$ which generates that of the operator $A^{(2)}$ as $A^{(2)}=A^{(2,0)}\bigoplus A^{(2,1)}$ and consequently $\exp[A^{(2)}]=\exp[A^{(2,0)}]\bigoplus\exp[A^{(2,1)}]$ clarifies the picture. However, instead of being in the space $\hhc$  the harmonic oscillator acts in we go a step further in modeling the action of the squeeze operator. More precisely, we duplicate the model in a way which is parallel to correspondence: ``configuration space'' $\mapsto$ an analogue of the  Segal-Bargmann space. Here ``$\mapsto$'' has the appearance of a kind of Segal-Bargmann transform; another occasion when the Segal-Bargmann transform appears in connection of squeezed states  is in a recent paper \cite{ali}.
\begin{center}
\boxed{\text{In all what follows $\lambda=\ulamek14$ governs   $A^{(2,0)}$ while  $\lambda=\ulamek34$ does   $A^{(2,1)}$.}}
\end{center}

\subsubsection{$A^{(2,i)}$ as Jacobi operators in $\llc^{2}(\rrb)$}  From the normalized  polynomials $p^{(\lambda)}_{n}$ already  given by  \eqref{1.19.02} we pass to the  {\em Meixner-Pollaczek functions }
\begin{equation*}
\mathfrak{p}^{(\lambda)}_{n}(x) \okr \big\vert\Gamma(\lambda+\I x)\big\vert p^{(\lambda)}_{n}(x),
\end{equation*}
which, due the formula (9.7.2) in \cite{koe},  satisfy the following orthogonality relation
\begin{equation*}
\int_{-\infty}^{\infty} \mathfrak{p}^{(\lambda)}_{n}(x) \overline{\mathfrak{p}^{(\lambda)}_{m}(x)} \D x = \delta_{n, m}.
\end{equation*}
The sequence $(\mathfrak{p}^{(\lambda)}_{n})_{n=0}^{\infty}$ is therefore orthonormal in
 $\mathcal{L}^{2}(\mathbb{R})$, each of the operators $A^{(2,i)}_{\xi}$, $i=0,1$, acts in the Hilbert space $\clolin(\mathfrak{p}^{(\lambda)}_{n})_{n=0}^{\infty}$ which is a subspace of $\llc^{2}(\rrb)$.

\subsubsection{$A^{(2,i)}$ as  multiplication operators in the space of the Segal-Bargmann type}\label{sb}Let us introduce the Hilbert space $\mathcal{H}_{\lambda}[\mathbb{C}; \nu(|z|) \D z]$, with 
\begin{equation}\label{4.11.13-6}
\nu(r) = \left[\frac{2^{\lambda}}{\pi\Gamma(2\lambda)}\right]^{2} r^{2\lambda-1} K_{2\lambda-1}(2 r),\quad r\in[0,+\infty).
\end{equation}
The basis in  $\mathcal{H}_{\lambda}[\mathbb{C}; \nu(|z|) \D z]$ is formed by the monomials
\begin{equation}\label{4.11.13-3}
\varPhi_{\lambda, n}(z) \okr 2^{-\lambda} \sqrt{2\pi}\, \Gamma(2\lambda) \frac{(-\I z)^{n}}{\sqrt{n! \Gamma(n + 2\lambda)}}, \quad z\in\mathbb{C}.
\end{equation}
\begin{lem}\label{4.11.13-4}
\begin{equation}\label{4.11.13-5}
\int_{\mathbb{C}} \varPhi_{\lambda, n}(z) \overline{\varPhi_{\lambda, m}(z)} \nu(|z|) \D z = \delta_{n, m}, \quad z\in\ccb,
\end{equation}
\end{lem}
\begin{proof}
Lemma \ref{4.11.13-4} can be shown by substituting formulas \eqref{4.11.13-3} and \eqref{4.11.13-6} into \eqref{4.11.13-5} and using \cite[vol.2, formula (2.16.2.2), p. 343]{APPrudnikov} \footnote{$\int_{0}^{\infty} x^{\alpha-1}K_{\nu}(cx) \D x = 2^{\alpha-2} c^{-\alpha} \Gamma(\ulamek{\alpha+\nu}{2}) \Gamma(\ulamek{\alpha-\nu}{2})$. By the way, the formula (3.26)  in \cite{Barut} is incorrect  which is irrelevant for the rest of that paper.} 
\end{proof}

 \subsubsection*{From $\mathcal{L}^{2}(\mathbb{R})$ to  $\mathcal{H}_{\lambda}[\mathbb{C}; \nu(|z|) \D z]$ {\em \`a la} Segal-Bargmann}
 $\mathcal{H}_{\lambda}[\mathbb{C}; \nu(|z|) \D z]$ is a reproducing kernel Hilbert space with the  kernel calculated for $\varPhi_{\lambda, n}$ as 
\begin{align*}
K^{(\lambda)}(t, \tau) & = \sum_{n=0}^{\infty} \varPhi_{\lambda, n}(t) \overline{\varPhi_{\lambda, n}(\tau)} = \frac{2\pi}{2^{2\lambda}} \Gamma^{2}(2\lambda) \sum_{n} \frac{(-\I t)^{n} (\I\bar{\tau})^{n}}{n! \Gamma(n + 2\lambda)} \\
& = \frac{2\pi}{2^{2\lambda}} \Gamma^{2}(2\lambda) (t\bar{\tau})^{-\lambda+1/2} I_{2\lambda-1}(2\sqrt{t\bar{\tau}}),\quad t,\tau\in\ccb.
\end{align*}
where $I_{\alpha}$ is the modified Bessel function of the first kind. This is one more kernel from which, applying the procedure developed in \cite{ho1,ho2}, one may get a new class of coherent states.

Let us find the unitary mapping of $\mathcal{L}^{2}(\mathbb{R})$ onto  $\mathcal{H}_{\lambda}[\mathbb{C}; \nu(|z|) \D z]$. We start with the formula\,\footnote{\label{poh}
\;{\em Pochhammer symbol} called sometimes  {\em shifted factorial} (see \cite[p. 4]{koe}) employed here.} for the generating function for the Meixner-Pollaczek polynomials, see \cite[formula (9.7.12)]{koe}:
\begin{align*}
\sum_{n=0}^{\infty} \frac{(-\I z)^{n}}{(2\lambda)_{n}} P^{(\lambda)}_{n}(x; \ulamek{\pi}{2}) &= \E^{z} {_{1}F_{1}}(\lambda+\I x; 2\lambda; -2z)  = z^{-\lambda} M_{-\I x, \lambda -1/2}(-2 z), \quad z\in\mathbb{C},
\end{align*} 
where $M_{\sigma, \nu}$ is the Whittaker function\label{wit}. Expressing $P^{(\lambda)}_{n}(x, \ulamek{\pi}{2})$ in terms of $\mathfrak{p}^{(\lambda)}_{n}(x)$ we get
\begin{align*}
\sum_{n=0}^{\infty} \varPhi_{\lambda, n}(z) \overline{\mathfrak{p}^{(\lambda)}_{n}(x)} &= \E^{z} \big\vert\Gamma(\lambda+\I x)\big\vert {_{1}F_{1}}(\lambda+\I x; 2\lambda; -2z) = (2 z)^{-\lambda} \big\vert\Gamma(\lambda+\I x)\big\vert M_{-\I x, \lambda-1/2}(-2z)
\end{align*}
and determine the transformation $\ggf_{\lambda}$ which sends $ \mathfrak{p}^{(\lambda)}_{n}$ to $\varPhi_{\lambda, n}$ as an integral one
\begin{equation}\label{4.11.13-10}
\varPhi_{\lambda, n}(z) = (\ggf_{\lambda} \mathfrak{p}^{(\lambda)}_{n})(z) = \int_{-\infty}^{\infty} G_{\lambda}(\bar{x}, z) \mathfrak{p}^{(\lambda)}_{n}(x) \D x
\end{equation}
with the kernel
\begin{equation*}
G_{\lambda}(\bar{x}, z) = (2 z)^{-\lambda} \big\vert\Gamma(\lambda+\I x)\big\vert M_{-\I x, \lambda-1/2}(-2z), \quad x\in\mathbb{R}, \quad z\in\mathbb{C}.
\end{equation*}
\begin{lem}\label{4.11.13-13}
\begin{equation}\label{4.11.13-14}
\int_{-\infty}^{\infty} G_{\lambda}(\bar{x}, z) \overline{G_{\lambda}(\bar{x}, z)} \D x = \frac{2\pi}{2^{2\lambda}}  \Gamma^{2}(2\lambda) |z|^{-2\lambda + 1} I_{2\lambda-1}(2|z|) = K^{(\lambda)}(z, \bar{z}). 
\end{equation}
\end{lem}
\begin{proof}
Using   \cite[vol. 3, formulae (7.2.2.15), p. 435 and (2.19.28.2), p. 211]  {APPrudnikov}
   we get
\begin{align*}\begin{split}
\int_{-\infty}^{\infty} G_{\lambda}(\bar{x}, z) \overline{G_{\lambda}(\bar{x}, z)} \D x &=  (2 r)^{-2 \lambda} \int_{-\infty}^{\infty} \big\vert\Gamma(\lambda+\I x)\big\vert^{2} M_{-\I x, \lambda-1/2}\, (-2z) M_{\I x, \lambda-1/2}(-2\bar{z}) \D x \\
& = (2 r)^{-2 \lambda} \E^{\I\pi\lambda}\int_{-\infty}^{\infty} \big\vert\Gamma(\lambda+\I x)\big\vert^{2} M_{\I x, \lambda-1/2}(2 z) M_{\I x, \lambda-1/2}(-2 \bar{z}) \D x \\
& = \pi (2 r)^{-2 \lambda + 1} \Gamma^{2}(2\lambda) \lim_{\rho\to 0} \frac{\exp[-(z-\bar{z})\tanh(\rho)]}{\cosh(\rho)} \I^{2\lambda-1}J_{2\lambda-1}\big(\ulamek{2\I r}{\cosh(\rho)}\big) \\
& = \frac{2\pi}{2^{2\lambda}}  \Gamma^{2}(2\lambda) r^{-2\lambda + 1} \I^{2\lambda-1}J_{2\lambda-1}(2\I r) = \frac{2\pi}{2^{2\lambda}}  \Gamma^{2}(2\lambda) r^{-2\lambda + 1} I_{2\lambda-1}(2r).\qedhere \end{split}
\end{align*}
\end{proof}

\subsubsection*{From  $\mathcal{H}_{\lambda}[\mathbb{C}; \nu(|z|) \D z]$ back to $\mathcal{L}^{2}(\mathbb{R})$; unitarity of $\ggf_{\lambda}$.}

This can be proved by showing that range of $\ggf_{\lambda}$ is dense in  $\mathcal{H}_{\lambda}[\mathbb{C}; \nu(|z|) \D z]$. Let us take the function $f_{q}\in\mathcal{L}^{2}(\mathbb{R})$, of the form $f_{\bar{q}}(w) = \overline{G_{\lambda}(\bar{q}, w)}$. Then \eqref{4.11.13-10} and \eqref{4.11.13-14} gives
\begin{align*}
\int_{-\infty}^{\infty} G_{\lambda}(\bar{q}, z) f_{\bar{q}}(w) \D q = \int_{-\infty}^{\infty} G_{\lambda}(\bar{q}, z) \overline{G_{\lambda}(\bar{q}, w)} \D q = K^{(\lambda)}_{\bar{w}}(z).
\end{align*}
The functions $K^{(\lambda)}_{\bar{w}}(z)$ are complete in $\mathcal{H}_{\lambda}$, this finishes the proof of unitarity of $G_{\lambda}$.

Define the operator $\funk W{ \mathcal{H}_{\lambda}[\mathbb{C}; \nu(|z|) \D z}{\mathcal{L}^{2}(\mathbb{R})}$ by
\begin{equation*}
(Wf)(x) = \int_{\mathbb{C}} \overline{G_{\lambda}(\bar{x}, w)} f(w) \nu(|w|) dw, \quad w\in\ccb.
\end{equation*}
For $g = \ggf_{\lambda}(W f) \in\mathcal{H}_{\lambda}$, we get
\begin{align*}
g(z) &= \int_{-\infty}^{\infty} G_{\lambda}(x, \bar{z}) \left[\int_{\mathbb{C}} \overline{G_{\lambda}(\bar{x}, w)} f(w) \nu(|w|) \D w\right] \D x 
= \int_{\mathbb{C}} f(w) \left[\int_{-\infty}^{\infty} G_{\lambda}(x, \bar{z}) \overline{G_{\lambda}(\bar{x}, w)} \D x\right] \nu(|w|) \D w \\
& = \int_{\mathbb{C}} f(w) K_{\lambda}(\bar{z}, w) \nu(|w|) \D w = f(z)
\end{align*}
which means $W$ is the inverse of $\ggf$.

The integral kernel corresponding to $G^{-1}_{\lambda}$ is given by
\begin{align*}
G^{-1}_{\lambda}(\bar{x}, z) = \overline{G_{\lambda}(\bar{x}, z)} &= (2\bar{z})^{-\lambda} \big\vert\Gamma(\lambda-\I x)\big\vert M_{\I x, \lambda-1/2}(-2\bar{z}) 
= (2\bar{z})^{-\lambda} \big\vert\Gamma(\lambda+\I x)\big\vert M_{\I x, \lambda-1/2}(-2\bar{z}) .
\end{align*}

\subsubsection*{The image of the  Jacobi operator $A^{{(2)}}_{\xi}$ in  $\mathcal{H}_{\lambda}[\mathbb{C}; \nu(|z|) \D z]$}

\begin{lem}\label{4.11.13-18}
\begin{equation}\label{mult}
G_{\lambda} \left(x\mathfrak{p}^{(\lambda)}_{n}(x)\right)(z) = \I\frac{1-z^{2}}{2 z} \varPhi_{\lambda, n}(z),\quad z\in\ccb.
\end{equation}
\end{lem}
\begin{proof}
Using the recurrence relation \eqref{2.19.02} and \eqref{4.11.13-3}, we get
\begin{align*}
G_{\lambda} \left(x\mathfrak{p}^{(\lambda)}_{n}(x)\right)(z) & = \int_{-\infty}^{\infty} G_{\lambda}(\bar{x}, z) x\mathfrak{p}^{(\lambda)}_{n}(x) \D x \\
& = \int_{-\infty}^{\infty} G_{\lambda}(\bar{x}, z) \left[\ulamek{1}{2}\sqrt{(n+1)(n+2\lambda)} \mathfrak{p}^{(\lambda)}_{n+1}(x) + \ulamek{1}{2}\sqrt{n(n+2\lambda-1)} \mathfrak{p}^{(\lambda)}_{n-1}(x)\right] \D x \\
& = \ulamek{1}{2}\sqrt{(n+1)(n+2\lambda)} \varPhi_{\lambda, n+1}(z) +  \ulamek{1}{2}\sqrt{n(n+2\lambda-1)} \varPhi_{\lambda, n-1}(z) \\
& = \ulamek{-\I z}{2} \varPhi_{\lambda, n}(z) + \ulamek{\I}{2 z} \varPhi_{\lambda, n}(z) = \ulamek{\I}{2} (\ulamek{1}{z} -z) \varPhi_{\lambda, n}(z).\qedhere
\end{align*}\end{proof}

The  $\mathcal{H}_{\lambda}[\mathbb{C}; \nu(|z|) \D z]$ parallels of the other members of the harmonic oscillator family  can be  derived as follows. From \eqref{4.11.13-3} it can be proved that ``creation'' and ``annihilation'' operators with respect to the basis $(\varPhi_{\lambda, n})_{n=0}^{\infty}$ act as
\begin{equation}\label{5.11.13-1}
-\I z \varPhi_{\lambda, n}(z) = \sqrt{(n+1)(n+2\lambda)} \varPhi_{\lambda, n+1}(z), \quad  \I\frac{\D}{\D z} \varPhi_{\lambda, n}(z) = \sqrt{\frac{n}{n+2\lambda-1}} \varPhi_{\lambda, n-1}(z).
\end{equation}
Formulas \eqref{5.11.13-1} imply (``number'' operator)
\begin{equation*}
z \frac{\D}{\D z} \varPhi_{\lambda, n}(z) = n \varPhi_{\lambda, n}(z)
\end{equation*}
which shares  the appearance with that of  the classical Segal-Bargmann space.

The counterparts of \eqref{5.11.13-1} in $\llc^{2}(\rrb)$ can be get from unitarity of $G_{\lambda}$ and formulas \eqref{4.11.13-10} and \eqref{5.11.13-1}
\begin{align*}
G^{-1}_{\lambda}\left[-\I z\varPhi_{\lambda, n}(z)\right] =\sqrt{(n+1)(n+2\lambda)} \mathfrak{p}^{(\lambda)}_{n+1}(x),\quad G^{-1}_{\lambda}\left[\I \frac{\D}{\D z} \varPhi_{\lambda, n}(z)\right] = \sqrt{\frac{n}{n+2\lambda-1}} \mathfrak{p}^{(\lambda)}_{n-1}(x).
\end{align*}

\subsubsection*{Symmetricity of the multiplication operator of \eqref{mult}}

The operator $M_{\lambda}$ of multiplication\,\footnote{\;Notice that the formula for $M_{\lambda}$ is the same regardless the space it acts in; though both the space  as well as the domain depend on $\lambda$.}  by $ \I\frac{1-z^{2}}{2 z}$ in $\mathcal{H}_{\lambda}[\mathbb{C}; \nu(|z|) \D z]$  is, according to \eqref{mult}, an image of the symmetric operator (the Jacobi one)
it must be necessary symmetric too. The multiplication by a rational function with a pole at $0$ acting in a space of entire functions  may look strange at a first glance though our reasoning does not leave any doubt. However, just for disbelievers we add an alternative,  direct argument for this a little bit amazing fact.

Symmetricity of $M_{\lambda}$ means
\begin{equation*}
\int_{\mathbb{C}} \ulamek{\I}{2}\Big(\ulamek{1}{z} - z\Big) \Phi_{\lambda, n}(z) \overline{\Phi_{\lambda, m}(z)} \nu(|z|) \D z = \int_{\mathbb{C}} \Phi_{\lambda, n}(z)\, \overline{\ulamek{\I}{2}\Big(\ulamek{1}{z} - z\Big)  \Phi_{\lambda, m}(z)} \nu(|z|) \D z.
\end{equation*}
Passing to polar coordinates under the integral and using explicitly \eqref{4.11.13-6} and \eqref{4.11.13-3} gives the above equality.

\section{Concluding remarks}\label{fhsz:s1}
We have proposed a precise solution of an intriguing problem of possible generalization of higher order squeezing. As we have already pointed out in the very introduction the existing so far attempts do  not explain satisfactorily why there is a disparity between the case $k=1,2$ and that of $k\Ge 3$. What is hidden behind is the fact that a Hilbert space operator can not live without its domain being explicitly manifested. The example we have in mind is a symmetric operator and selfadjoint as well, in which case the domain makes the difference (cf. \cite{kon} or for much more particularities also \cite{fabio}). This is invisible when the notion of a Hermitian operator is the only in use, with a consequence of an automatic transplantation of the ${^{\dagger}}$ operation from finite matrices together with its algebraic properties to would-be Hilbert space operators. The typical argument: $U=\exp(\pm\I H)$ is unitary if $H$ is Hermitian, i.e. $H=H^{\dagger}$, is far from being correct as long as $H$ is not (essentially) selfadjoint -  the case of $A_{\xi}^{(3)}$ makes a strong warning here. Therefore some caution even for the trivially looking cases of $A_{\xi}^{(1)}$ and $A_{\xi}^{(2)}$ has to be undertaken  - this is  a  message  our universal approach conveys.  
Not taking into account behaviour of domains may  result in serious, troublesome problems as the paper \cite{gal} inquires into.

Although we have focused ourselves  on mathematical aspects of higher order squeezing  the paper sends also a clear message to physicists:  generalizing naively  squeezing operators to higher order fails because the out-coming operators do not obey fundamental quantum mechanical requirements postulated by von Neumann - they are "ill-defined" in the physical jargon. Moreover, this fact is by no means restricted to squeezing circumstances.  
As we
have already emphasized  the same situation one faces if studying the
$k$-photon Rabi model. The $3$-photon Rabi model has been recently suggested
\cite{mondloch} to explain the mechanism of phase locking through the spontaneous
three-photon scattering which, if confirmed experimentally, allows us to
conjecture that "ill-defined" phenomenological description may be cured in a
mathematically rigorous way and to achieve this one should look for
selfadjoint  extensions of the $3$-photon Rabi interaction.

Summing up, though our main goal has been to prove rigorously impossibility of generalizing squeezing to higher orders in a naive way, one of the benefits of our investigations is to call reader's attention to a need of being aware  how important and helpful a domain is for studying  properties of a specific operator; for the thorough discussion of the issue the chapter \cite{fabio} highly recommended.

\section{aknowledgment}
 {The authors are extremely grateful  to the  referees for their deep insight into the paper which benefited in its final version.
The third author was supported by the MNiSzW grant NN201 546438.}

\section*{Appendix. Laborious though indispensable calculations}
{\sc Proof of Lemma \ref{3Oct13-1}.}
For \eqref{3Oct13-1a}, Pochhammer symbol appears again, cf. footnote  \footnoteref{poh},
 use induction as follows\begin{align*}
& (A^{(k,i)})^{n+1} e^{(k,i)}_{p} = \sum_{r=0}^{n} \prod_{s=1}^{r} \sum_{j_{s}=s-1}^{j_{s-1}+1} \frac{(1+i+(p-r+j_{s}))_{k}}{[(1+i+pk)_{(n-2r)k}]^{-1/2}}  (-\I\E^{\I\theta})^{n-2r} A^{(k,i)} e^{(k,i)}_{p+n-2r} \\
&\quad = \sum_{r=0}^{n} \prod_{s=1}^{r} \sum_{j_{s}=s-1}^{j_{s-1}+1} \frac{(1+i+(p-r+j_{s}))_{k}}{[(1+i+pk)_{(n+1-2r)k}]^{-1/2}} (-\I\E^{\I\theta})^{n+1-2r} e^{(k,i)}_{p+n+1-2r} \\
&\quad + \sum_{r=0}^{n} \prod_{s=1}^{r} \sum_{j_{s}=s-1}^{j_{s-1}+1} \frac{(1+i+(p-r+j_{s}))_{k} (1+i+pk)_{(n-2r)k}}{\sqrt{(1+i+pk)_{(n-1-2r)k}}} (-\I\E^{\I\theta})^{n-1-2r} e^{(k,i)}_{p+n-1-2r}.
  \end{align*}
For \eqref{3Sep13-4} apply the Pythagorean law to \eqref{3Oct13-1a}.

To prove the right hand side of the inequality \eqref{1.29.01} it is enough to show $\| (A^{(k, i)})^n e^{(k, i)}_{p} \|^{2} \Le k^{n+1} [i + (p+n)k]!/(i+pk)!$. 
Indeed,
\begin{gather*}
\| (A^{(k, i)})^n e^{(k, i)}_{p} \|^{2} \Le (1+i+pk)_{nk} + \sum_{r=1}^{n} \frac{[(1+i+(p+n-2r)k)_{k}]^{2}}{(n-r+1)^{-1}} \sum_{j_{2}=1}^{n-r+1} [(1+i+(p-r+j_{2})k)_{k}]^{2} \\  
 \times \prod_{s=3}^{r} \sum_{j_{s}=s-1}^{j_{s-1}+1} [(1+i+(p-r+j_{2})k)_{k}]^{2} (1+i+pk)_{(n-2r)k} \Le \ldots \Le (1+i+pk)_{nk}\\ + \sum_{r=1}^{n} \left[\frac{(1+i+pk)_{(n-r-1)k}}{(1+i+pk)_{(n-2r)k}}\right]^{2}
 (n-r+1)^{r-1} \sum_{j_{s}=r-1}^{n-1} \left[(1+i+(p-r+j_{s})k)_{k}\right]^{2}
   \Le (1+i+pk)_{nk}\\ + k^{2} \left[\frac{(1+i+pk)_{(n-1)k}}
 {(1+i+pk)_{(n-2)k}}\right]^{2} (1+i+pk)_{(n-2)k}
 = (1+i+pk)_{nk} \left[1 + k^{n}  - \frac{k^{n+1}}{i+(p+n)k}\right]\\ \Le (1+k^{n}) (1+i+pk)_{nk} \Le 2k^{n}  (1+i+pk)_{nk}.
\end{gather*}
The left hand side of the inequality  \eqref{1.29.01} can be automatically get from \eqref{3Sep13-4}, because 
\begin{equation*}
\| (A^{(k, i)})^n e^{(k, i)}_{p}\|^{2}\geq (1+i+pk)_{nk}. 
\end{equation*}

 \subsubsection*{Further inequalities}For $p=0, 1, \ldots$ and $k\Ge 3$ we have
\begin{equation}\label{1.4.2}
\frac{(i+pk)!}{\sqrt{(i+pk-k)! (i+pk+k)!}} + \sqrt{\frac{(i+pk)!}{(i+pk+k)!}} - 1 < 0, \quad i=0, 1, \ldots.
\end{equation}

{\sc Proof of \eqref{1.4.2}}. Indeed, going with the left hand side of \eqref{1.4.2} on
\begin{align*}
&\prod_{j=0}^{k-1} \left(\frac{kp-j+i}{kp+j+1+i}\right)^{\frac{1}{2}} + \prod_{j=0}^{k-1} \frac{1}{\sqrt{kp+j+1+i}} - 1 < \left(\frac{kp+i}{kp+1+i}\right)^{k/2} + \left(\frac{1}{kp+1+i}\right)^{k/2} - 1 \\
&  < 1 - \frac{k}{2(kp+1+i)} + \frac{1}{kp+1+i} - 1 < - \frac{k-2}{2(kp+1+i)} 
\end{align*}
which  makes \eqref{1.4.2}; here the assumption  $k\Ge 3$ is \underbar{essential}.

 For $p=0, 1, \ldots$, $k=1, 2, \ldots$, and $i=0, 1, \ldots, k-1$
\begin{equation} \label{2.4.2}
\frac{(i+pk-k)!}{\sqrt{(i+pk)! (i+pk-2k)!}} - \frac{(i+pk)!}{\sqrt{(i+pk-k)! (i+pk+k)!}} < 0.
\end{equation}

{\sc Proof of \eqref{2.4.2}}.
Proceeding as in the proof of \eqref{1.4.2} we get
\begin{align*}
&\prod_{j=0}^{k-1} \left(\frac{kp-k-j+i}{kp-k+j+1+i}\right)^{\frac{1}{2}}  - \prod_{j=0}^{k-1} \left(\frac{kp-j+i}{kp+j+1+i}\right)^{\frac{1}{2}}  < \left(\frac{kp-k+i}{kp-k+1+i}\right)^{k/2} -  \left(\frac{kp+i}{kp+1+i}\right)^{k/2} \\
& < \left(1 - \frac{1}{kp-k+1+i}\right)^{k/2} - \left(1 - \frac{1}{kp+1+i}\right)^{k/2}\
 \\
 &=\left[\left(1 - \frac{1}{kp-k+1+i}\right)^{k/2} + \left(1 - \frac{1}{kp+1+i}\right)^{k/2}\right]^{-1}\left[ \left(1 - \frac{1}{kp-k+1+i}\right)^{k} - \left(1 - \frac{1}{kp+1+i}\right)^{k}\right]
\\
&\Le\frac1{\sqrt 2}\left[ \left(1 - \frac{1}{kp-k+1+i}\right) - \left(1 - \frac{1}{kp+1+i}\right)\right]\,\sum_{l=0}^{k-1}\left(1 - \frac{1}{kp-k+1+i}\right)^{l}  \left(1 - \frac{1}{kp+1+i}\right)^{k-l-1}<0.
\end{align*}
which implies \eqref{2.4.2}.

\bibliographystyle{amsplain}

\end{document}